\documentclass[twocolumn,
aps,superscriptaddress,
amsmath,amssymb,
pra
]{revtex4-2}

\usepackage{graphicx}
\usepackage{dcolumn}
\usepackage{bm}

\usepackage[utf8]{inputenc}
\usepackage{color}
\usepackage{ bbold }
\usepackage[svgnames]{xcolor}
\usepackage{amsthm}
\usepackage{algorithm}
\usepackage{algpseudocode}
\usepackage{ulem}
\usepackage[colorlinks,breaklinks]{hyperref}
\usepackage{titlesec}
\usepackage[toc,page]{appendix}

\hypersetup{
	bookmarksnumbered,
	pdfstartview={FitH},
	citecolor={darkgreen},
	linkcolor={darkred},
	urlcolor={darkblue},
	pdfpagemode={UseOutlines}}
\definecolor{darkgreen}{RGB}{50,190,50}
\definecolor{darkblue}{RGB}{0,0,190}
\definecolor{darkred}{RGB}{238,0,0}

\newcommand{\ket}[1]{\left\vert#1\right\rangle}
\newcommand{\bra}[1]{\left\langle #1 \right\vert}

\newcommand{\ketbra}[2]{| #1\rangle \langle #2|}

\newcommand{\Tr}{\text{Tr}}

\newcommand{\de}[1]{\left(#1\right)}

\newcommand{\DE}[1]{\left\{{#1}\right\}}
\newtheorem{theorem}{Theorem}

\renewcommand{\emph}[1]{\textit{#1}}
\newcommand{\comment}[1]{}
\newcommand{\cv}[1]{\textcolor{black}{#1}}

\begin{document}

\preprint{APS/123-QED}

\title{Quantum key distribution overcoming extreme noise:\\simultaneous subspace coding using high-dimensional entanglement}

\author{Mirdit Doda}
\affiliation{Institute of Physics, Slovak Academy of Sciences, 845 11 Bratislava, Slovakia}%
\affiliation{Institute for Quantum Optics and Quantum Information - IQOQI Vienna, Austrian Academy of Sciences, Boltzmanngasse 3, 1090 Vienna, Austria}%
\author{Marcus Huber}
\affiliation{Institute for Quantum Optics and Quantum Information - IQOQI Vienna, Austrian Academy of Sciences, Boltzmanngasse 3, 1090 Vienna, Austria}%
  \affiliation{Institute for Atomic and Subatomic Physics, Vienna University of Technology, Vienna, Austria}
  \affiliation{Vienna Center for Quantum Science and Technology, Atominstitut, TU Wien,  1020 Vienna, Austria}%
\author{Gláucia Murta}
\affiliation{Institut für Theoretische Physik III, Heinrich-Heine-Universität Düsseldorf,Universitätsstraße 1, D-40225 Düsseldorf, Germany}
\author{Matej Pivoluska}
\affiliation{Institute of Physics, Slovak Academy of Sciences, 845 11 Bratislava, Slovakia}%
\affiliation{Institute of Computer Science, Masaryk University, 602 00 Brno, Czech Republic}%
\author{Martin Plesch}
\affiliation{Institute of Physics, Slovak Academy of Sciences, 845 11 Bratislava, Slovakia}%
\affiliation{Institute of Computer Science, Masaryk University, 602 00 Brno, Czech Republic}%
\author{Chrysoula Vlachou}
\affiliation{Instituto de Telecomunicações, Av. Rovisco Pais 1, 1049-001 Lisboa, Portugal}
\affiliation{Departmento de Matemática, Instituto Superior Técnico,
Universidade de Lisboa, Av. Rovisco Pais 1, 1049-001 Lisboa, Portugal}

\date{\today}

\begin{abstract}
High-dimensional entanglement promises to increase the information capacity of photons and is now routinely generated exploiting spatio-temporal degrees of freedom of single photons. A curious feature of these systems is the possibility to certify entanglement despite strong noise in the data. We show that it is also possible to exploit this noisy {high-dimensional} entanglement {for quantum key distribution} by introducing a  protocol that uses mutliple subspaces of the high-dimensional system simultaneously. Our protocol can be used to {establish a secret key}  even in {{extremely noisy} experimental} conditions, where {qubit protocols fail}. {To show that, we analyze the performance of our protocol} for noise models that apply to the two most commonly used sources of high-dimensional entanglement: time bins and spatial modes.
\end{abstract}

\maketitle
\section{Introduction}
Quantum communication is one of the most mature areas of quantum technologies, with Quantum Key Distribution (QKD) \cite{BENNETT20147,E91, RevModPhys.81.1301,Lo2014,2019arXiv190601645P} as its most prominent example.
QKD is a cryptographic primitive allowing two parties, Alice and Bob, to securely establish a shared secret key in the presence of an adversary, Eve.

One of the primary scientific challenges in the transition to commercial QKD applications is the relatively low key rate and the high susceptibility to noise. Contrary to classical communication with light, the information encoded in quantum states can neither be copied nor amplified,  meaning that only few photons survive long distance transmissions.
Furthermore, single photons are challenging to detect and hard to isolate, leading to a lot of noise in the data.

{The problem of successfully performing a QKD protocol over noisy channels is {also} connected to {the} basic task of distributing entanglement. Essentially, if the channel is too noisy to distribute entanglement, it is also too noisy for QKD \cite{RevModPhys.81.1301,2019arXiv190601645P}.}
{This correspondence points to a possible solution to the above challenges --- the use of QKD protocols which utilize high-dimensional (HD) entanglement.
HD entanglement is known to feature  high resistance to noise according to theoretical noise models (e.g. white noise).
At the same time, HD states can encode more bits per photon.} {Another reason for considering entanglement-based QKD protocols is that they are less prone to practical attacks \cite{Pracattack} compared to their prepare-and-measure counterparts, thus also providing higher practical security. } {\cv {Indeed, using entanglement based systems one can relinquish trust in the source and place it in the hands of an untrusted node. Furthermore, entanglement distribution is a necessary step towards (partially) device independent implementations. }}{{Currently}, HD QKD protocols are becoming practical, because HD  quantum states of entangled photons can be routinely produced in the lab} using temporal \cite{PhysRevLett.64.2495,PhysRevLett.88.013601,PhysRevA.57.3123,Molina-Terriza2007,PhysRevLett.94.100501,WALBORN201087,Lima:11,Karimi2014,Rubinsztein_Dunlop_2016,Llewellyn_2019,Bavaresco_2018,Schaeff_2015,Schneeloch_2019,2020arXiv200404994H,Gmez2020MultidimensionalEG,2020arXiv200409964H}, frequency  encoding \cite{PhysRevA.88.032322,PhysRevLett.103.253601} or multiple ones simultaneously in a so-called hyperentangled state \cite{doi:10.1080/09500349708231877,PhysRevLett.95.260501,Vergyris_2019,Imany2019,PhysRevA.101.032302}.

Unfortunately, bringing this idea to practice is not straightforward, as the physical nature of the carriers and the actual noise become very important. HD quantum systems are not as easily controlled and measured as polarisation qubits, and all implementations come with their own limitations and additional sources of noise. While 
used for proof-of-principle experiments in QKD before \cite{AGS03,PhysRevA.82.030301,PhysRevA.88.032309,Brougham_2013,Mirhosseini_2015,Ding2017,PhysRevA.87.062322,Lee2015,PhysRevA.91.022336,2016arXiv161101139L,nik:alb:05,nik:ran:alb:06,vla:kra:mat:pau:sou:18,cha:15,Chau2,Chau3,bec:tit:00,cer:bou:kar:gis:02,Gr_blacher_2006,Sit:17,Fickler_2020,2020arXiv200403498V,Etcheverry2013,PhysRevA.96.022317}, they have never been competitive to regular qubit encoding, {and the majority of practical QKD implementations still uses binary encoding of quantum states in photons, such as polarisation \cite{xu2019secure} or time-bin qubits \cite{timebinqkd}.}  {In fact,} perhaps surprisingly, {even} the theoretically predicted higher noise resistance of HD entanglement has only recently been demonstrated in {\emph{realistic scenarios}} \cite{Ecker_2019}. There, the data obtained from measuring HD states distributed over realistic very noisy channels can be used to certify the presence of entanglement.
{However, such noisy data} are not necessarily useful for QKD.
This is because correcting the errors {on the outcomes obtained using multiple-outcome measurements is more demanding than correcting binary outcomes and it conventionally comes at a cost that can obliterate the advantage of using HD entanglement.}
{Hence, the question whether such noisy HD entanglement can actually be useful remains open.}
In other words, `\emph{can we still harness the HD nature of entanglement in situations where noise dominates the signal and qubit-based QKD would be impossible?'}
In this work we affirmatively answer that question and provide {an entanglement-based HD QKD protocol} with \emph{simultaneous subspace coding}. We provide detailed noise models for two paradigmatic implementations of HD entanglement to showcase the suitability of our protocol for practical advantages in QKD.

\section{The protocol}
\begin{algorithm}[H]
	\floatname{algorithm}{Protocol}
	\caption{Subspace QKD}\label{prot}
	\begin{algorithmic}[1]
\State\label{step:distribute}\noindent\textbf{Distribution.} A source distributes a state $\rho_{AB}$ to Alice and Bob.
\State\label{step:M}\noindent\textbf{Measurement.} Upon receiving the state, Alice and Bob choose  independently at random bits $w_A$ and $w_B$, respectively, such that $p(w_A=1)=p(w_B=1)=\varepsilon \ll 1$. If $w_A=0$, Alice performs a measurement in the $\{A_1^x\}_{x=0}^{d-1}$ basis, otherwise she measures in $\{A_2^x\}_{x=0}^{d-1}$. Similarly, Bob measures his part of the state in $\{B_1^y\}_{y=0}^{d-1}$ or $\{B_2^y\}_{y=0}^{d-1}$, accordingly. They record the outcomes $x$ and $y$ in the register $X$ and $Y$ respectively. Steps $1$ and $2$ are repeated $N$ times.
\State\label{step:AD}\noindent\textbf{Sifting and Subspace selection.} Through a classical public authenticated channel, Alice and Bob reveal for each iteration their basis choice and the values $m_A$ and $m_B$ of the subspaces their outcomes belong to.
\begin{itemize}
    \item If $m_A=m_B=m$ and $w_A=w_B$,   Alice and Bob set $M=m$, $x'=x-mk$ and $y'=y-mk$.
    \item If $m_A\neq m_B$ or $w_A\neq w_B$, they set $M=\perp$, $x'=y'=\perp$, and discard the round
\end{itemize}
\State\label{step:PE}\noindent\textbf{Parameter estimation.}
Alice and Bob use the second measurement basis outcomes (\emph{test rounds}) and some of the key first measurement basis outcomes (\emph{generation rounds}) to estimate correlations for each block $M=m$. The remaining measurement results form the raw key.
\State\label{step:IR-PA}\noindent\textbf{Information reconciliation and Privacy amplification.} Alice and Bob proceed with information reconciliation and privacy amplification in each subspace and extract the final key.
	\end{algorithmic}
\end{algorithm}

The general idea is to use a $d\times d$ dimensional entangled quantum system to perform multiple instances of a QKD protocol simultaneously in non-overlaping subspaces.
The choice of subspaces is arbitrary, but to simplify the notation we formulate the protocol using subpaces of equal size $k$. 
The protocol requires two measurement settings -- the computational basis measurement and a measurement in a basis mutually unbiased with respect to the computational in each subspace of size $k$.
Therefore, let $\{A_1^x\}_{x=0}^{d-1}$ and $\{B_1^y\}_{y=0}^{d-1}$ denote the {projectors on the} computational basis of $\mathcal{H}_A$ and $\mathcal{H}_B$, respectively, and
$\{A_2^x\}_{x=0}^{d-1}, \{B_2^y\}_{y=0}^{d-1}$  {a tensor product of} projectors on the mutually unbiased basis vectors in subspaces of size $k$ {(see Appendix \ref{app:RateProof} for details)}. 
Alice's measurement outcome $x = mk+i$ is interpreted as outcome $i$ in the $m$-th subspace. Bob's measurement outcome $y$ is interpreted analogously.


%
%
%
%
%
The protocol consists of the steps described in Protocol \ref{prot}.
To assess the {protocol's efficiency}, we calculate the achievable key rate $K$ in the asymptotic limit.
Initially, we assume that Eve is restricted to \emph{collective attacks}, i.e., at each iteration she attacks identically and independently of the previous, and she can perform a measurement on her ancillary system at any future time.
This assumption can be dropped later with a de Finetti-type argument \cite{deFinetti1,deFinetti2}  and thus security against any \emph{general coherent attack} can be obtained. {\cv{Note that since we derive our results in the asymptotic setting, we can safely use the de Finetti theorem without degrading the key-rate expression.}}
Furthermore, we assume that Eve prepares the entangled states that Alice and Bob share.
Proving security under this assumption implies security for any implementation of the protocol. {\cv {Besides these security assumptions, we also assume a detector model with fair sampling, access of the legitimate parties to randomness and to a classical public authenticated channel, as usual in QKD protocols.}}
Under these assumptions the asymptotic key rate is given by \cite{2005RSPSA.461..207D,RennerThesis}:
       $K\geq H(X|E_T)-H(X|Y),$
where $H(X|E_T)$ is the von Neumann entropy of Alice's key round outcome $X$ conditioned on the total information available to the eavesdropper Eve at the end of Step \ref{step:PE}, given that Eve holds a purification of the  state $\rho_{AB}$, and $H(X|Y)$ is the conditional Shannon entropy between Alice's and Bob's key-round outcomes.

The asymptotic key rate of Protocol \ref{prot} is given by
 \begin{align}\label{eq:rateSubQKD}
K_{TOT} \geq \sum_{m=0}^{\ell-1} P(M=m)K_m,
\end{align}
where $P(M=m)$ is the probability that both Alice and Bob obtain an outcome in subspace $m$ and $K_m$ is the corresponding rate, given by $K_m=H(X'|E_T)_{\tilde{\rho}^m}-H(X'|Y')_{\tilde{\rho}^m}$, where $\tilde{\rho}^m$ is the state effectively shared by the parties in the subspace $m$.
The proof of \eqref{eq:rateSubQKD} is based on similar techniques as the ones used in protocols with advantage distillation \cite{Mau93,GL03,KBR07,BA07}. For the definition of the state $\tilde{\rho}^m$ and the detailed proof of this result, see Appendix \ref{app:RateProof}.

To compute $K_m$ for the subspace $m$ we lower bound the conditional entropy $H(X'|E_T)_{\tilde{\rho}^m} $ by the conditional min-entropy:
 $H(X'|E_T)_{\tilde{\rho}^m}\geq H_{min}(X'|E_T)_{\tilde{\rho}^m}$, where $H_{min}(X'|E_T)_{\tilde{\rho}^m}=-\log P_g^m$ and $P_g^m$ is the average probability that Eve can guess Alice's outcome computed using the effective state $\tilde{\rho}^m$.
 To determine $P_g^m$, for a subspace $m$ from measurement results, we use the correlations of Alice's and Bob's outcomes in the second basis, expressed as \begin{equation}
     W_k^m=\sum_{i=0}^{k-1} P(ii|22,m),
 \end{equation}
where $P(ii|22,m)=\frac{P(x=mk+i,y=mk+i|22)}{P(M=m)}$ is the probability that Alice and Bob obtain equal outcomes when they measure in $\{A_2^x\}_{x=0}^{d-1}$ and $\{B_2^y\}_{y=0}^{d-1}$ and obtain outcomes in the subspace $m$.
 In Appendix \ref{app:TheWitness} we present and solve the optimization problem that enables us to show that Eve's guessing probability for the subspace $m$ can be  expressed as
 a function of the subspace dimension $k$ and the correlation $W_k^m$ as:
\begin{equation}\label{eq:Pg}
   H_{min}(X'|E_T)_{\tilde{\rho}^m}\!=\!-\!\log_2\!\! \left(\!\tfrac{\left(\sqrt{W_k^m}+\sqrt{(k-1)(1-W_k^m)}\right)^2}{k}\!\right)\!\!.
\end{equation}{}

The conditional entropy $H(X'|Y')_{\tilde{\rho}^m}$ for subspace $m$ can be estimated {directly from the measurement outcomes} in the basis $\{A_1^x\}_{x=0}^{d-1}$ and $\{B_1^y\}_{y=0}^{d-1}$.

 Note that a similar technique of subspace encoding has been used previously to encode qubits using HD systems \cite{cha:15,Chau2}. However, our protocol is intrinsically different. Our construction allows to explore subspaces of arbitrary sizes and the encoding is simultaneous in all the subspaces. Moreover, our protocol, as well as the noise analyses presented in the following Sections, are designed for an entanglement-based implementation, which offers an higher level of  practical security compared to prepare and measure protocols, whose noise robustness and feasibility were considered in \cite{Chau3}.

\emph{Isotropic state example}: First we investigate a paradigmatic noise model, i.e. a {maximally entangled} state mixed with white noise, $\rho_d(v) = v\ketbra{\psi_d^+}{\psi_d^+} + \frac{(1-v)}{d^2}\mathbb{1}_{d}\otimes\mathbb{1}_{d}$ with visibility $v$.
{We} calculate the asymptotic key rate for $k=d$, i.e. Alice and Bob use standard QKD and derive the key from the {full} Hilbert space.
{Measuring $\rho_d(v)$ {in} the second {basis} for both Alice and Bob leads to}
$W_d=v+\frac{1-v}{d}$ and thus via equation \eqref{eq:Pg}
$H(X|E_T) \geq -\log_2 \left(\frac{ \left( \sqrt{vd+1-v}+(d-1)\sqrt{1-v} \right)^2 }{d^2}\right)$.
To determine $H(X|Y)$, we observe that,
given $\rho_d(v)$, the probability distribution for Alice obtaining result $x$ and Bob obtaining result $y$ in the key rounds is given by $P_{key}(xy)=
v\frac{\delta_{xy}}{d}+\frac{1-v}{d^2},$ and the respective conditional probability distribution is $P_{key}(x|y)=v\delta_{xy}+\frac{1-v}{d}$.

Similar calculations can be done
in case Alice and Bob perform Protocol \ref{prot} with subspaces   $\mathcal{H}_{A_m}\otimes\mathcal{H}_{B_m}$ of size $k\times k$.
In such a case they effectively measure the state $\rho_k^m(\tilde{v})$ in each subspace $m$, which can obtained by
projecting $\rho_d(v)$ onto this subspace.
Because of the symmetry of $\rho_d(v)$, the state $\rho_k^m(\tilde{v})$ is independent of $m$ and its density matrix is equivalent to $\rho_k(\tilde{v})= \tilde{v} \ketbra{\psi^+_k}{\psi^+_k} + \frac{(1-\tilde{v})}{k^2}\mathbb{1}_{k}\otimes\mathbb{1}_{k} $,
where $\ket{\psi^+_k}=\frac{1}{\sqrt{k}}\sum_{i=0}^{k-1}\ket{ii}$ and
$\tilde{v} = \tilde{v}(d,v,k):=vd/(vd+k-vk)$.

For each subspace  $H(X'|E_T)_{\tilde{\rho}^m}$ and $H(X'|Y')_{\tilde{\rho}^m}$ we can now set $\tilde{\rho}^m=\rho_k(\tilde{v})$.  Measurements of this state {in the second basis} lead to $W_k^m=\frac{vd+1-v}{vd+k-vk}$, which (using Eq. \eqref{eq:Pg}) results in:
\begin{equation}
H(X'|E_T)_{\tilde{\rho}^m}\geq
-\log_2\left(
\tfrac{ \left( \sqrt{vd+1-v}+(k-1)\sqrt{1-v} \right)^2 }{k(vd+k-vk)}\right).\label{eq:pgsub}\end{equation}
Evaluating $H(X'|Y')_{\rho_k(\tilde{v})}$ and summing over all subspaces leads to:
\begin{align}
     K&_{\text{TOT}}^{iso}(d,v,k)
\!\geq\!
\left(\tfrac{vd+k-vk}{d}\right)
\log_2\!\left(\!\!\tfrac{k}
    {\left(\sqrt{vd+1-v}+(k-1)\sqrt{1-v}\right)^2}
    \!\!\right)
\nonumber\\
&+\!(\tfrac{vd+1-v}{d})\!\log_2(vd\!+\!1\!-\!v) \!+\!\tfrac{(k-1)(1-v)}{d}
\!\log_2(1\!-\!v)
.
 \label{keyiso}
 \end{align}
For each $d$ and $v$, which are known experimental parameters, one can optimize $K_{\text{TOT}}^{iso}(d,v,k)$ over the subspace size $k$ to determine the protocol implementation {with} the optimal key rate.
Another interesting quantity is the critical visibility, i.e. the visibility beyond which one cannot obtain a positive key rate anymore. This is generally a complicated function of  $k$ and $d$. However, considering even $d$ and $k=2$, it can be shown that the key rate is positive for $v>\tfrac{1}{1+0.0893d}$.
{For} constant $v$, one can always obtain a positive key rate by increasing the global dimension $d$. This is in accordance with the {previously observed fact} -- the robustness of entanglement in the isotropic state increases with the dimension.
In practice, however, the $v$ is not a constant, but rather a function of $d$, {and strongly depends} on the particular implementation.
To {infer} whether our protocol actually holds the potential to outperform qubit-based  protocols, it is thus essential  {to} take the experimental {specifications} into account.
In what follows we study two different state-of-the-art implementations of our protocol with dimension-dependent noise models.
The first employs temporal and the second spatial degrees of freedom {of photons} for the generation of HD entanglement. The motivation for using these particular setups is that  they were recently shown to provide an advantage for entanglement certification \cite{Ecker_2019}.

\section{Realistic noise models}
\subsection{Implementation using temporal degrees of freedom}
To start with the temporal implementation we consider a hyper-entangled state of the form
$\ket{\Psi}= \ket{\phi^-}_{AB} \otimes\int dt f(t) \ket{t}_A \otimes \ket{t}_B$.
The  part of the state entangled in energy-time is produced by a laser source via spontaneous parametric down-conversion, and the interference of photons in the temporal domain is subsequently enabled by introducing entanglement in polarization by means of $\ket{\phi^-}_{AB}$.
Alice and Bob measure the time of arrival, $t$, of the photons, i.e., the time when the detectors click. The maximum resolution with which they can detect photons arriving at the same time is given by the duration of a time bin, $t_b$, and with respect to it, they determine a time frame, $F$, outside of which any photon arriving is considered ``lost''.
They choose $F=d t_b\text{ for }d\in \mathbb{N}$, effectively discretizing the energy-time space to obtain a space of dimension $d$; the encoding space.
The frames in which they both had one click are post-selected and used for the key rate.
In Appendix \ref{app:Temporal} we present in detail a noise model for this setup, in which we take into account noise effects due to the interaction of the photons with the environment and due to the imperfect detectors.
In particular, photons might be lost before arriving to the lab, and other photons coming from the environment might enter and make the detectors click.
The environmental photons are the main source of noise for this implementation.
Moreover, {we consider} dark counts, i.e., detector clicks in the absence of a photon, and finally, {that} the detectors might not click in the presence of a photon.
We can, then, express the key rate as a function of the dimension $d$  and the visibility $v$, which is the probability that, given that both Alice and Bob had one click, this click is due to a photon coming from the laser source and not due to an environmental photon or a dark count. For this model the visibility is given as
\begin{equation}
    v(d)=1/(1+dt_bT_AT_B \gamma^{-1}),
\end{equation}
and the production rate of post-selected frames$/sec$  is
\begin{align}
R(d,v)=e^{-d\, t_b (T_A+T_B+\gamma)}(d\,t_bT_A T_B + \gamma),
\end{align}
 where $T_{A/B},\gamma$ are experimental parameters incorporating all quantities that are independent of $d$. $T_{A/B}$ is the average number of uncorrelated clicks per second, coming from dark counts, environmental photons or laser photons when one of the parties is affected by losses or detectors' inefficiencies; $\gamma$ is the average number of detected entangled photons per second, i.e, the photons  coming from the laser source that were not lost and produced a click. The achievable key rate expressed in $bits/sec$ is  $K(d)=R(d,v)K^{iso}_{\text{TOT}}(d,v,k)$. {\cv {In practical implementations this number is further multiplied by $(1-\varepsilon)^2$, i.e. the probability that both Alice and Bob used the first measurement. Since here we are dealing with asymptotic key rate, we can choose $\varepsilon$ arbitrarily close to $0$, hence we disregard it.}}
 In Figure \ref{fig:temporal}, we plot $K(d)$ versus the noise-to-signal ratio for different subspace-dimension choices $k$ in various total dimensions $d$.

 \begin{figure}[h]
     \centering
     \includegraphics[width= 0.47\textwidth]{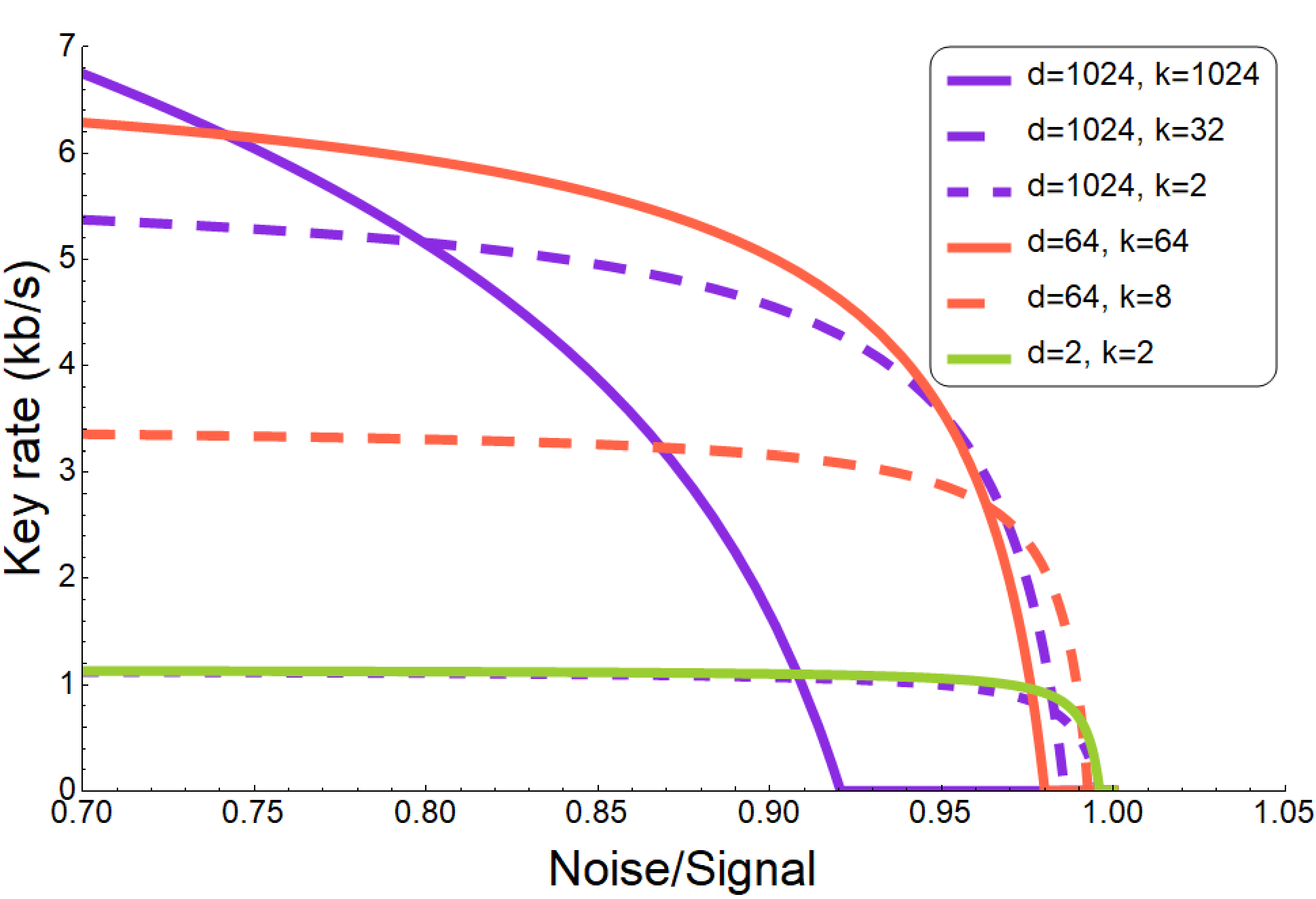}
     \caption{
    Achievable key rate versus noise-to-signal ratio for the temporal domain.
    The noise-to-signal ratio is the average number of non-entangled photons that arrive in the lab (including singles and taking into account detector inefficiencies) divided by the overall average number of clicks per second, assuming that these quantities are the same for both parties. For this implementation this quantity is dimension-independent.  Since we vary the frame sizes, but not the bin sizes, the end points are independent of the total dimension and the optimal noise resistance is always achieved in two-dimensional subspaces, \cv{which includes the traditional qubit encoding ($d = 2, k = 2$)}. The plot illustrates the fact that choosing $k>2$, however, can significantly increase the total key rate at lower noise levels. \cv{Another fact one can directly observe is that the best key rate for given noise level can be obtained by fine-tuning values of $d$ and $k$.}}
     \label{fig:temporal}
 \end{figure}

\subsection{Implementation using spatial degrees of freedom}
We move to photons entangled in spatial degrees of freedom.
Due to spatial symmetry, the state produced by the laser source {is} of the form $\ket{\Psi}=\sum_{l=-\infty}^{\infty}c_l\ket{l}_A\ket{-l}_B$, where $l$ denotes momentum modes and $c_l$ depends on the source specifications. This state is subsequently projected in a space spanned by a finite subset of modes, $l$, with cardinality $d$; our encoding space, hence arises from an effective discretization with respect to the finite resolution of the detectors. In our {noise} model, we consider noise effects originating from losses, environmental photons, detector inefficiencies and dark counts.
For the key rate, Alice and Bob post-select the rounds in which they both obtained one click, and just like in the previous implementation, the visibility $v$ includes the rounds where the clicks came from a source photon pair. In this implementation each party needs a detector \emph{for each mode}, resulting in dark counts contributing the most to noise through more frequent accidental coincidences. We calculate the visibility to be
\begin{equation}
v(d)=\frac{e^{ \tfrac{\gamma}{d}} -1}{e^{ \tfrac{\gamma}{d}} -1+d\big[1-e^{- (\mu^A+\tfrac{\xi^A}{d})}\big] \big[1-e^{-(\mu^B+\tfrac{\xi^B}{d})}\big] }, \end{equation}
and the rate of post-selected rounds$/sec$  is given by
 \begin{align}
 R&(d,v)= C  e^{-d(\mu^A+\mu^B)}e^{\tfrac{(\xi^A+\xi^B)}{d}}\\\nonumber
 &\!\!\!\!\times\Bigg\{\!d^2\Big[1\!-\!e^{- \mu^A-\tfrac{\xi^A}{d}}\Big]\! \Big[1\!-\!e^{-\mu^B-\tfrac{\xi^B}{d}}\Big]\!+\!d\left( e^{ \tfrac{\gamma}{d}}\!-\!1\right)\!\!\Bigg\},
 \end{align}
  where $\mu^{A/B}$ is the average number of dark counts per detector, $\xi^{A/B}$ is the average number of uncorrelated photons due to the environment, losses and detector inefficiencies, $\gamma$ is the average number of detectable correlated photons, and finally $C$ is a related, also dimension-independent, parameter (for details see Appendix \ref{app:Spatial}). We can now express the achievable key rate in $bits/sec$ as $K(d)=R(d,v)K_{\text{TOT}}^{iso}(d,v,k)$.
  In Figure \ref{fig:spatial}, we plot $K(d)$ versus the total dimension $d$ for different choices of  subspace size $k$.
\begin{center}
\begin{figure}[h]
    \centering
    \includegraphics[width = 0.47\textwidth]{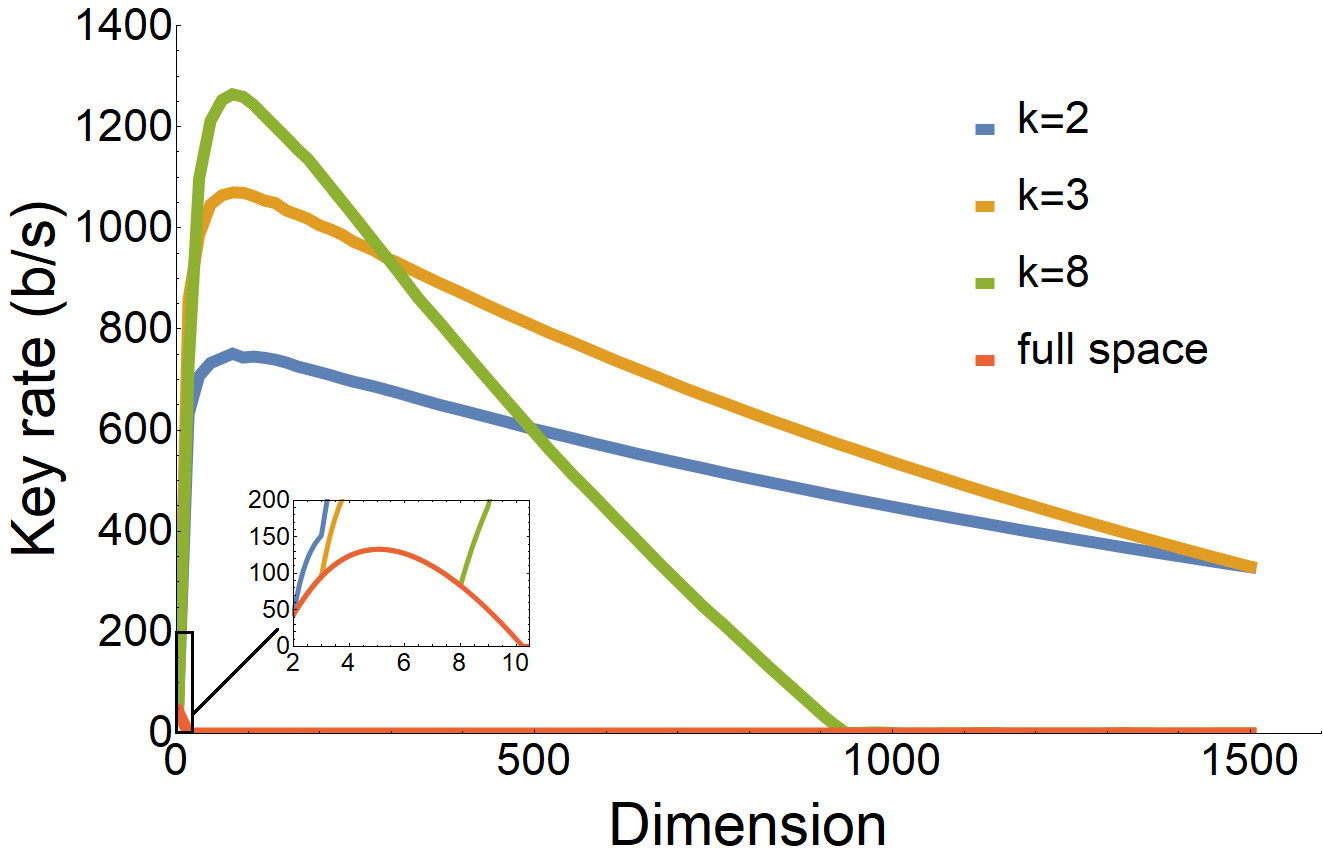}
    \caption{Achievable key rates versus dimension for different subspace encodings in the spatial domain. The parameters that we use here are: detector efficiencies $P_C=60\%$, losses for one party $P_L=98,4\%$, dark counts $\mu=600$ clicks/s, environmental photons $\nu=21000$  photons/s, coincidence window $\Delta t=10^{-7} s$, average laser photons (at source) $\lambda=200000$ photons/s. Increasing the Hilbert space by adding detectors imposes a natural limit beyond which the extra dimensions are not useful, but in fact detrimental. Also, different subspace sizes are optimal for different total dimensions. Except for small dimensions, the subspace encoding tremendously increases key rates over traditional full encoding.}
    \label{fig:spatial}
\end{figure}
\end{center}
\section{Conclusions} We presented a simultaneous subspace coding entanglement-based {HD} QKD protocol. Using two noise models for the most paradigmatic platforms for photonic HD entanglement, we showcase that the protocol can indeed provide a viable pathway towards practically improved QKD.
Most of the improvement comes from the fact that {encoding in} subspaces of HD systems {are} {is} more resilient to physical noise \cite{Ecker_2019}, compared to directly encoding in a comparable dimension.
It seems counter-intuitive that a qubit subspace reaches the highest noise resistance, when the entire premise is that HD systems are inherently more resistant to noise. The principal reason is that for the most common experimental implementations, the entanglement of the HD state is also present in two-dimensional subspaces, where the cost of error correction is the lowest. So, while larger subspaces feature more entanglement, smaller subspaces have a smaller error correction overhead, leading to this intricate interplay between noise, rate and subspace dimension that we observe. And in the cases of extremal noise they are the smallest subspaces that yield the highest (or any) key.  Surprisingly, the optimal subspace size in noisy scenarios often goes beyond two dimensions.
The actual value of achievable key rates highly depends on the implementation and specific noise parameters, from dark counts and background to losses and device fidelities. For all parameter ranges, HD encodings led to improved key rates. We believe that our noise models together with the SDPs for computing key rates for our protocol will be useful for optimizing system parameters for a broad family of future setups. {\cv {The security analysis and the noise models can be further refined and adjusted accordingly to account for other sources of noise and different noise regimes. For instance, at the noise regime that we consider, multi-photon detection events are negligible, therefore we discarded them without affecting the security of the protocol. However, in different noise regimes simply discarding these events might open up a security loophole and the security analysis should account for them, e.g. by treating them as additional noise (decreasing the key rate) depending on the implementation.}} \cv{Last but not least, the theoretical predictions for the achievable key rates presented in this paper were recently successfully verified in a proof of principle implementation of our protocol using photons entangled in path degree of freedom \cite{2020arXiv201103005H}. }
\vspace{\baselineskip}
\begin{acknowledgments}
The authors would like to thank Mateus Ara\'{u}jo for his comments. MH acknowledges funding from the Austrian Science Fund (FWF) through the START project Y879-N27.
CV acknowledges support from the Belgian Fonds de la Recherche Scientifique  -- FNRS, under grant no R.50.05.18.F (QuantAlgo). The QuantAlgo project has received funding from the QuantERA ERA-NET Cofund in Quantum Technologies implemented within the European Union's Horizon 2020 Programme. CV also acknowledges support from the Fundação para a Ciência e a Tecnologia (FCT) through national funds, by FEDER, COMPETE 2020, and by the Regional Operational Program of Lisbon, under UIDB/50008/2020 (actions QuRUNNER,
QUESTS) and QuantumMining POCI-01-0145-FEDER-031826. GM is funded by the Deutsche Forschungsgemeinschaft (DFG, German Research
Foundation) under Germany's Excellence Strategy – Cluster of Excellence
Matter and Light for Quantum Computing (ML4Q) EXC 2004/1 – 390534769.
MD, MPi and MPl acknowledge funding from VEGA project 2/0136/19. MPi and MPl additionally acknowledge GAMU project MUNI/G/1596/2019.
\end{acknowledgments}
\bibliography{biblio}{}

\onecolumngrid
\newpage
\begin{center}
    \huge{Appendix}
\end{center}
\appendix

\section
{Key rates of the subspace QKD protocol}\label{app:RateProof}
Here we prove the key rate expression for the Subspace-QKD Protocol \ref{prot}.
The key rate can be computed using the key rates of each subspace, as stated in the following theorem.
\begin{theorem}
The asymptotic key rate of the subspace-QKD protocol, Protocol \ref{prot},  is given by
\begin{align}
   K_{TOT}\geq \sum_{m=0}^{\ell -1}p(M=m) \de{H(X'|E)_{\tilde{\rho}^m}-H(X'|Y')_{\tilde{\rho}^m}},
\end{align}
where the conditional entropies are evaluated on the states $ \tilde{\rho}^m_{X'Y'E}$ given by
\begin{align}
    \tilde{\rho}^m_{X'Y'E}=(\mathcal{E}^{\mathcal{M}}_{X'Y'\leftarrow AB}\otimes id_E)(\ket{\tilde{\psi}_{ABE} ^{m}})
\end{align}
and $\ket{\tilde{\psi}_{ABE} ^{m}}$ is the purification of the state
\begin{align}
   \rho_{AB}^{m}= \frac{\Pi^m_A\otimes \Pi^m_B ( \rho_{AB}) \Pi^m_A\otimes \Pi^m_B}{p(M=m)}
\end{align}
with
\begin{align}
  p(M=m)={\rm Tr} (\Pi^m_A\otimes \Pi^m_B \rho_{AB}).
\end{align}
\end{theorem}

\begin{proof}
The protocol explores multiple subspaces of size $k$, where $d=\ell\cdot k$. Then both $\mathcal{H}_A$ and $\mathcal{H}_B$ can be divided into $\ell$ subspaces of size $k$ as
$\mathcal{H}_A = \mathcal{H}_{A_0}\oplus\dots\oplus
\mathcal{H}_{A_{\ell-1}}$ and $\mathcal{H}_B = \mathcal{H}_{B_0}\oplus\dots\oplus
\mathcal{H}_{B_{\ell-1}}.$

The QKD protocol involves two measurement settings: $\{A_1^x\}_{x=0}^{d-1}$ and $\{B_1^y\}_{y=0}^{d-1}$ denote the computational bases of $\mathcal{H}_A$ and $\mathcal{H}_B$, respectively, and
$\{A_2^x\}_{x=0}^{d-1}, \{B_2^y\}_{y=0}^{d-1}$ are sums of mutually unbiased measurements in subspaces of size $k$. Formally, $A_2^x = UA_1^xU^\dagger$ and $B_2^y=U^*B_1^yU^\top$, where $U = \sum_{m=0}^{\ell-1}\sum_{i,j=0}^{k-1} \omega_k^{ij}\ket{mk+i}\bra{mk+j}$ and
$\omega_k = e^{\frac{2\pi i}{k}}$. Note that the two measurements of the parties are block diagonal with blocks of size $k$. Therefore,
Alice's measurement outcome $x = mk+i$ is interpreted as outcome $i$ in the $m$-th subspace. Bob's measurement outcomes are interpreted analogously.

The subspace-QKD protocol, given in in the main text, can be divided in the following maps
\begin{align}
    \mathcal{E}^{SubspaceQKD}_{K_AK_BE''\leftarrow ABE}(\rho_{ABE})=\mathcal{E}^{IR-PA}_{K_AK_BE''\leftarrow X'Y'E'}\circ\mathcal{E}^{Sub}_{X'Y'E'\leftarrow XYE}\circ ( \mathcal{E}^{\mathcal{M}}_{XY\leftarrow AB}\otimes id_E)(\rho_{ABE}).
\end{align}
where $\mathcal{E}^{\mathcal{M}}_{XY\leftarrow AB}\otimes id_E$ represents the measurements implemented in Step \ref{step:M}, $\mathcal{E}^{Sub}_{X'Y'E'\leftarrow XYE}$ corresponds to the subspace selection implemented in Step \ref{step:AD}, and finally
$\mathcal{E}^{IR-PA}_{K_AK_BE''}$ describes the classical post-processing applied to the raw key consisting of information reconciliation and privacy amplification.

The difference between the subspace-QKD protocol, Protocol \ref{prot}, and a standard high-dimensional QKD protocol is in the subspace selection described in Step~\ref{step:AD}.
The selection of subspaces, Step~\ref{step:AD} in Protocol \ref{prot}, corresponds to Alice and Bob applying a local projection described by  $\DE{\Pi^m}_{m=0}^{\ell-1}$ with elements
\begin{align}
    \Pi^m=\sum_{i=0}^{k-1}\ketbra{mk+i}{mk+i},
\end{align}
and selecting the cases where both get the same outcome. This step resembles the procedure called \emph{advantage distillation} that has been studied in the classical setting \cite{Mau93} as well as in QKD protocols \cite{GL03,KBR07,BA07}. Indeed the subspace selection in Step \ref{step:AD} aims to select the rounds in which higher correlations between Alice and Bob are observed. This is exactly the objective of advantage distillation procedures such as the ones studied in \cite{GL03,KBR07,BA07}. However the procedures studied in \cite{GL03,KBR07,BA07} consist of processing several rounds together (corresponding to a joint quantum operations in several copies of the state) while the procedure given by Step \ref{step:AD} only involve one copy of the state (and corresponds to a single copy filter operation).

Now, note that the map implemented in Step~\ref{step:AD} is a projection into subspaces, and this operation commutes with both measurements of Alice and Bob. Therefore we can, instead, describe the QKD protocol as first Alice and Bob perform the projection and then proceed with the measurements in the resulting state
\begin{align}
    \mathcal{E}^{SubspaceQKD}_{K_AK_BE''\leftarrow ABE}(\rho_{ABE})=\mathcal{E}^{IR-PA}_{K_AK_BE''\leftarrow X'Y'E'}\circ( \mathcal{E}^{\mathcal{M}}_{X'Y'\leftarrow A'B'}\otimes id_E)\circ \mathcal{E}^{Sub}_{A'B'E'\leftarrow ABE}(\rho_{ABE}).
\end{align}

So in order to perform the security analysis we will use this alternative description. Our goal is to determine the state shared by the parties after the action of the map $\mathcal{E}^{Sub}_{A'B'E'\leftarrow ABE}$. Initially, Eve distributes a state $\rho_{AB}$ to Alice and Bob such that she holds a purification of it
\begin{align}
    \rho_{AB}={\rm Tr}_E\ketbra{\psi_{ABE}}{\psi_{ABE}}
\end{align}
Applying the map $\mathcal{E}^{Sub}_{A'B'E'\leftarrow ABE}(\ket{\psi}_{ABE})$ leads to
\begin{align}
     \mathcal{E}^{AD}_{A'B'E'\leftarrow ABE}(\rho_{ABE})=\rho_{A'B'EM}
\end{align}
where the register $M$ records the result of the projection, i.e. $M$ takes value $m$ if $m_A=m_B=m$, and $M=\perp$ otherwise:
\begin{align}\label{eq:stateAD}
\rho_{A'B'EM}=\sum_{m=0}^{\ell -1}p(M=m)\rho_{ABE}^{m}\otimes\ketbra{m}{m}_M + p(M=\perp) \ketbra{\perp \perp}{\perp \perp}_{AB}\otimes \rho_E\otimes\ketbra{\perp}{\perp}_M
\end{align}
where
\begin{align}
   \rho_{ABE}^{m}= \frac{\Pi^m_A\otimes \Pi^m_B\otimes I_E ( \ketbra{\psi_{ABE}}{\psi_{ABE}}) \Pi^m_A\otimes \Pi^m_B\otimes I_E}{p(M=m)}
\end{align}
and $p(M=m)$ is the probability that Alice and Bob get outcome $m$ in the projection, given by
\begin{align}
  p(M=m)={\rm Tr} (\Pi^m_A\otimes \Pi^m_B\otimes I_E \ketbra{\psi_{ABE}}{\psi_{ABE}}).
\end{align}
Now, Alice and Bob will perform measurements on the state \eqref{eq:stateAD} in order to generate a key:
\begin{align}
    \rho_{X'Y'EM}=\mathcal{E}^{\mathcal{M}}_{X'Y'\leftarrow A'B'}\otimes id_{EM}(\rho_{A'B'EM})
\end{align}
The entropy of Alice's outcome after measuring the state ${\rho_{A'B'EM}}$, conditioned on the information available to the eavesdropper is given by
\begin{align}
    H(X'|EM)_{\rho_{X'EM}}&=\sum_{m=0}^{\ell -1}p(M=m) H(X|EM=m)_{\rho_{X'EM}}\\
    &=\sum_{m=0}^{\ell -1}p(M=m) H(X'|E)_{\rho_{X'E}^{m}}\\
    &\geq \sum_{m=0}^{\ell -1}p(M=m) H(X'|E)_{\tilde{\rho}_{X'E}^{m}}
\end{align}
The first equation follows from the properties of conditional von Neumann entropy for cq-states. In the last step we consider the entropy evaluated on the state that results from Alice and Bob applying the measurements to $\tilde{\rho}_{A'B'EM}$, where
\begin{align}\label{eq:purification}
\tilde{\rho}_{A'B'EM}=\sum_{m=0}^{\ell -1}p(M=m)\tilde{\rho}_{ABE}^{m}\otimes\ketbra{m}{m}_M + p(M=\perp) \ketbra{\perp \perp}{\perp \perp}_{AB}\otimes \rho_E\otimes\ketbra{\perp}{\perp}_M
\end{align}
and
\begin{align}
\tilde{\rho}_{ABE}^{m}={\rm Tr}_E\ketbra{\psi^m_{ABE}}{\psi^m_{ABE}}
\end{align}
where $\ket{\psi^m_{ABE}}$ is the purification of $\rho_{AB}^{m}$. Giving Eve the purification of $\rho_{AB}^{m}$ before the measurements only increases her power, which proves the lower bound.
Similarly for the required information to be exchanged for information reconciliation
\begin{align}
    H(X'|Y'M)_{\rho_{X'Y'M}}&=\sum_{m=0}^{\ell -1}p(M=m) H(X'|Y'M=m)_{\rho_{X'Y'M}}\\
    &=\sum_{m=0}^{\ell -1}p(M=m) H(X'|Y')_{\rho_{X'Y'}^{m}}\\
     &=\sum_{m=0}^{\ell -1}p(M=m) H(X'|Y')_{\tilde{\rho}_{X'Y'}^{m}}
\end{align}
and the last step follows from the fact that $\tilde{\rho}_{AB}^{{m}}=\rho_{AB}^{{m}}$.
\end{proof}

\section{Solution of the SDP and the choice of $W$}
\label{app:TheWitness}
In this appendix, we present in detail the optimization problem for calculating the average guessing probability of Eve, and its solution. The average guessing probability is obtained by maximizing, over all possible tripartite states $\rho_{ABE}$ (recall that Eve holds a purification of $\rho_{ABE}$) and all possible measurements of Eve $\{E^e\}_e$, the probability $P_g$ that Eve's guesses Alice's outcomes, and then performing a weighted average of these probabilities:
\begin{align} \label{SDPgeneral}
\begin{split}
P_{g}= 	&\hspace{0.5cm} \max_{  \rho_{ABE},\{E^e\}_{e=0}^{d-1}}\sum_{e,y} \Tr \left( \rho_{ABE} A_1^e \otimes  B_1^y \otimes E^e\right) \\
\text{s.t.} 	&\hspace{0.5cm} \Tr \left( \widehat{W} \rho_{AB}\right)=W,\\
	&\hspace{0.5cm} \rho_{ABE} \geq 0, \\
    &\hspace{0.5cm} \Tr \left(\rho_{ABE}\right)=1, \\
	&\hspace{0.5cm} E^e \geq 0 \hspace{0.2cm}\forall e,\\
	&\hspace{0.5cm} \sum_e E^e = \mathbb{1},\\
\end{split}
\end{align}
where $A_1,
B_1$ stand for the computational basis
and $\widehat{W}$ is the yet-to-be-defined operator, with $W$ its measured value that constrains the optimization. $W$ will be constructed depending on which \textit{target} state the experiment is trying to produce. 
Here, we consider $W$ to be the average value of equal outcomes that Alice and Bob get in the second basis, $W=\sum_{x} P(xx|22),$
which is the average value of the operator $\widehat{W}=\sum_x A_2^x \otimes B_2^x$.
Our goal now is to express the guessing probability as a function of $W$ and $d$, by solving the following optimization problem which we obtain from the previous one by substituting $\rho_e:=\Tr_{E} \left(\rho_{ABE} E^e\right)$:
\begin{align*}
\begin{split}
P_{g}(W,d)= 	&\hspace{0.5cm} \max_{\{\rho_e\}_{e}}\sum_{e=0}^{d-1} \Tr \left( \rho_e A_1^e \otimes \mathbb{1}_{d} \right) \\
\text{s.t.} 	&\hspace{0.5cm} \Tr \left( \widehat{W} \sum_{e=0}^{d-1} \rho_e\right)=W,\\
	&\hspace{0.5cm} \rho_e \geq 0\hspace{0.2cm},\forall e\in\{0,\ldots,d-1\}, \\
    &\hspace{0.5cm} \Tr \left(\sum_{e=0}^{d-1} \rho_e\right)=1, \\
\end{split}
\end{align*}
whose dual is
\begin{align*}
    \begin{split}
     \min_{\gamma,S} 	\hspace{0.2cm}	&\gamma+SW,\\
      \text{s.t.} 	\hspace{0.2cm}&\gamma \mathbb{1}_{d^2} + S \widehat{W} \geq \ketbra{e}{e} \otimes \mathbb{1}_{d}  \hspace{0.2cm},\forall e.
    \end{split}
\end{align*}
The non-zero eigenvalues of $S \widehat{W}-\ketbra{e}{e} \otimes \mathbb{1}_{d}$, as a function of $S$ and the local dimension $d$, for a given $e$ are
$\lambda_{\pm}=\frac{S-1 \pm \sqrt{(S-1)^2+4S(d-1)/d}}{2}$ (each one with degeneracy $d$) ,
 and the optimization now reads:
\begin{align*}
    \begin{split}
     \min_{\gamma,S} 	\hspace{0.2cm}	&\gamma+SW,\\
      \text{s.t.} 	\hspace{0.2cm}&\gamma+ \lambda_- \geq 0 \hspace{0.3cm}\forall e,\\
      &\gamma+ \lambda_+ \geq 0 \hspace{0.3cm}\forall e.
    \end{split}
\end{align*}
Since $\lambda_+ \geq \lambda_-$, for all $S,d$, and $\lambda_\pm$ are the same for all $e$, we can relax the constraints to:
\begin{align*}
    \begin{split}
     \min_{\gamma,S} 	\hspace{0.2cm}	&\gamma+SW,\\
      \text{s.t.} 	\hspace{0.2cm}&\gamma+ \lambda_- \geq 0.
    \end{split}
\end{align*}
We finally solve
$\frac{\partial}{\partial S}\lambda_-=W,$ which gives
\begin{equation*}
    W=\frac{1}{2}\left(1 - \frac{2S+2-4/d}{ 2\sqrt{(S-1)^2+4S(d-1)/d}}\right),\ \ \ \
    S=\frac{2}{d}-1+\frac{1-2W}{d}\sqrt{\frac{d-1}{W(1-W)}}\ \ \ \text{and}\ \ \ \
    \lambda_-=-\frac{d-1}{d}-\frac{W}{d}\sqrt{\frac{d-1}{W(1-W)}}.
\end{equation*}
and obtain the form of the guessing probability as a function of $W$ and $d$:
\begin{equation}
    P_g(W,d)=-\lambda_-+SW=
    \frac{\left(\sqrt{W}+\sqrt{(d-1)(1-W)}\right)^2}{d}.\nonumber
    \label{keysol}
\end{equation}

\section{Implementation with temporal degrees of freedom}\label{app:Temporal}
We start by considering a hyper-entangled state of the form
\begin{equation}
\ket{\Psi}= \ket{\phi^-}_{AB} \otimes\int dt f(t) \ket{t}_A \otimes \ket{t}_B.
\end{equation}
This is the state of two entangled photons (one for Alice and one for Bob), with two degrees of freedom: the time of arrival $t$ of the photon at the respective labs of Alice and Bob and their polarization.
The time of arrival, i.e., the time at which the detector clicks, is a continuous variable, which can be discretized by considering time bins of size $t_b$. Setting a time frame $F$ outside of which a photon is ``lost'' and taking $F$ to be a multiple of $t_b$ we have effectively a discrete system of dimension $d=\frac{F}{t_b}$.\\
The probability that at the frame $[0,F]$ exactly $n$ pairs of entangled photons are produced is given by the Poisson distribution:
\begin{equation*}
P_F(n)=\frac{(\lambda F)^n e^{-\lambda F}}{n!},
\end{equation*}
with $\lambda$ being the production rate of the photon pairs. Both $F$ and $\lambda$ are tunable parameters, and we assume that $\lambda$ is small enough, such that multi-photon events in the same frame are negligible. In particular, we choose $\lambda $ such that
\begin{equation*}
    P_F(n\geq 2)=1-P_F(0)-P_F(1)=1-(1+\lambda F)e^{-\lambda F}<\epsilon.
\end{equation*}
The average number of photons per frame is $\lambda F$ and for $\lambda F<0.2$ we get $P_F(n\geq2)<0.015$, which is small enough with respect to the noise scale in our model. Note, though, that by decreasing the production rate $\lambda$ we are also decreasing the key rate, see Equation (\ref{eq:rate_tb}) at the end of this section, therefore we should tune these parameters carefully.\\

We consider two types of noise, namely the noise due to the interaction of the photons with the environment before entering the labs of Alice and Bob, and the noise introduced due to the detectors' inefficiency.
 Because of its interaction with the environment, a photon can be lost with probability $P_L$.
Given $n$ photons, the probability that $n_L$ of them are lost while the rest arrive at the lab is
$
   P_L^{n_L} (1-P_L)^{n-n_L} \genfrac(){0pt}{2}{n} {n_L}.
$

Moreover, photons from the environment may be introduced in the system. We assume that the environment produces on average $\nu$ photons per second. The number of photons arriving to the frame from the environment will then follow the Poisson distribution,
$    P_E(n)=(\nu F)^n e^{-\nu F}/(n!).$

As far as the noise due to the detectors is concerned, each detector has probability $P_C$ to click when a photon arrives, and probability per second $\mu$ to click when no photon is there. These events are called dark counts and they also follow the Poisson distribution,
$ P_D(n)=(\mu F)^n e^{-\mu F}/(n!).$

The probability that both Alice and Bob receive in their labs $(i,j)$ photons in a time frame $F$ is
\begin{equation*}
    P(i,j)=\sum_{n=0}^{\infty}\sum_{n_1=\max\{n-i,0\}}^n \sum_{n_2=\max\{n-j,0\}}^n  P_F(n) P_L^{n_1} (1-P_L)^{n-n_1} {n \choose n_1} P_L^{n_2} (1-P_L)^{n-n_2} {n \choose n_2} P_E(i-n+n_1)P_E(j-n+n_2).
\end{equation*}

Given that $i$ photons enter, the probability of obtaining exactly one click in a frame $F$ is
\begin{equation*}
    P({\text{click}}|i)=(1-P_C)^{i-1}\left( P_C P_D(0)i+(1-P_C) P_D(1)\right)
    =e^{-\mu F}(1-P_C)^{i} \left(\frac{iP_C}{1-P_C}+\mu F\right),
\end{equation*}
and the probability that both Alice and Bob get one click is
$P(11)=\sum_{i,j=0}^\infty P({\text{click}}|i) P({\text{click}}|j) P(i,j).$\\
After applying the approximation $P_F(n\geq 2)\approx 0$, we can calculate $P(11)$ to be:
\begin{equation*}
    P(11)\approx e^{-F[2(\mu +\nu P_C)+\lambda]}F\beta,\end{equation*}
  with  $ \beta:=\Big[\lambda  \alpha^2 +  F(\mu+\nu P_C )^2\Big]$ and  $\alpha:=\big[P_C(1-P_L)+F(\mu +P_C  \nu) P_L   + F(\mu+\nu P_C)(1-P_C)(1-P_L) \big]$.\\
In the above expressions note that, if only a photon pair is produced,
the probability that it passes and gets detected is $P_C(1-P_L)$, the probability that it passes but does not get detected is $(1-P_C)(1-P_L)$, and the probability that it gets lost is $P_L$, and an environment photon or a dark count is making the click instead with probability $(F(\mu+\nu P_C))$. The term $e^{-2(\mu+\nu P_c)F}$ is the probability that all the extra photons of the environment are not detected and there are no dark counts.
If there is no pair in the frame, the click must have come from the environment or it is a dark count.

If the setup is asymmetric (one detector is close to the source, the other is far), we can modify the formula to include different parameters for Alice and Bob:
\begin{align*}
     P(11)&\approx e^{-(\mu^A+\mu^B+\nu^A P_C^A+\nu^B P_C^B+\lambda)F} F\Big[\lambda \alpha^A \alpha^B +F(\mu^A+\nu^A P^A_C )(\mu^B+\nu^B P^B_C )\Big].
\end{align*}
In our noise model, we do not consider finite size effects (the number of rounds is sufficiently large), neither border effects on the frame ($F$ is sufficiently large, so the error of the clock that decides when the frame begins and ends is negligible), nor errors related to the relaxation time of the detectors (which is the time a detector needs before being able to detect another photon. If the frame were approximately the same size as the relaxation time this effect would be important, but we choose $F$ to be sufficiently large for this purpose). We also assume that the interaction with the environment can only destroy a photon, and that the photon pairs coming from the environment are uncorrelated. Furthermore, $F$ and the production rate $\lambda$ are chosen such that the probability of observing two or more entangled photons during a single frame is negligible. With these assumptions, we have a model that gives us the rate of ``valid'' rounds per second, as a function of  $F$, which, in turn, is proportional to the local dimension $d$:
\begin{equation}
    R(d)=P(11)/F
    =e^{-F[2(\mu+\nu P_C)+\lambda]} \beta.
\end{equation}
For large $F=d\, t_b$, we have  $\beta\approx F(\mu+\nu P_C)^2[(1+P_C P_L - P_C)^2+\lambda F]\propto F.$

We can also estimate the visibility, i.e. the probability that -- given that both parties had exactly one click -- the photons that clicked were the entangled ones coming from the laser source and not the environment or dark counts. First, we calculate the probability of a photon pair to survive and get detected
\begin{equation*}
P_S=P_F(1)(1-P_L)^2 P^2_C e^{-2F(\mu+\nu P_C)},
\end{equation*}
which gives the visibility as a function of the dimension $d$ to be
$v(d)=P_S/P(11)   =\lambda  (1-P_L)^2 P^2_C/\beta,$ while for an asymmetric setup we have $
     v(d)=\lambda  (1-P_L^A)(1-P_L^B) P^A_C P^B_C/\beta.$\\
Finally, for large $d$, $\lambda\propto1/F$ and $\beta\propto F$, thus making the visibility scale as $d^{-2}$.

\vspace{\baselineskip}
In order to take into account multi-photon events, we  write $\alpha$  and $P(11)$ as
\begin{align*}
    \alpha(n)&=(1-P_C+P_C P_L)^{n-1}\big[nP_C(1-P_L) + F(\mu+\nu P_C)(1-P_C+P_C P_L) \big],\\
     P(11)&=e^{-(\mu^A+\mu^B+\nu^A P_C^A+\nu^B P_C^B)F} \sum_{n=0}^\infty P_F(n)\alpha^A(n)\alpha^B(n),
\end{align*}
and obtain \begin{equation}
    P(11)=e^{-F(T_A+T_B+\gamma)}(F^2T_A T_b + F\gamma)\ \ \ \ \text{and}\ \ \ \
      R=e^{-F(T_A+T_B+\gamma)}(FT_A T_b + \gamma), \label{eq:rate_tb}
\end{equation}
where $S=1-P_C+P_C P_L,\ Q=\mu+\nu P_C,\ T_{A/B}=Q^{A/B}+\lambda S^{B/A}(1-S^{A/B})\text{ and }\gamma=\lambda (1-S^A)(1-S^B)$.

Accordingly, the generalized success probability becomes
\begin{equation*}
    P_S=e^{-F(Q^A+Q^B)}\sum_{n=0}^\infty P_F(n) (S^A S^B)^{n-1}(1-S^A)(1-S^B)n
    =e^{-F(T_A+T_B+\gamma)}\gamma F.
\end{equation*}
The maximum of $P_S$ is for $F=1/(T_A+T_B+\gamma)$, and the visibility becomes
$v=1/(1+FT_AT_B \gamma^{-1}).$

\section{Implementation with spatial degrees of freedom}\label{app:Spatial}
A basic parameter of our model is $\Delta t$, the coincidence window in which two events, for Alice and Bob, are considered coincident. Note that multiple clicks in the same coincidence window are treated as a single event.
Another parameter is related to the projection of an infinite-dimensional entangled state of the form $
\ket{\Psi}=\sum_{l=-\infty}^{\infty}c_l \ket{l}_A\otimes\ket{-l}_B$ into a finite dimensional space. This is the probability $P_P(d):=\Tr\big( \mathbb{1}_{d^2} \ketbra{\Psi}{\Psi} \mathbb{1}_{d^2} \big)$, which we assume to be constant, thus providing lower dimensions with an advantage. One could give an advantage to higher dimensions by dropping this assumption.

\begin{figure}[H]
     \centering
     \includegraphics[width= 0.75\textwidth]{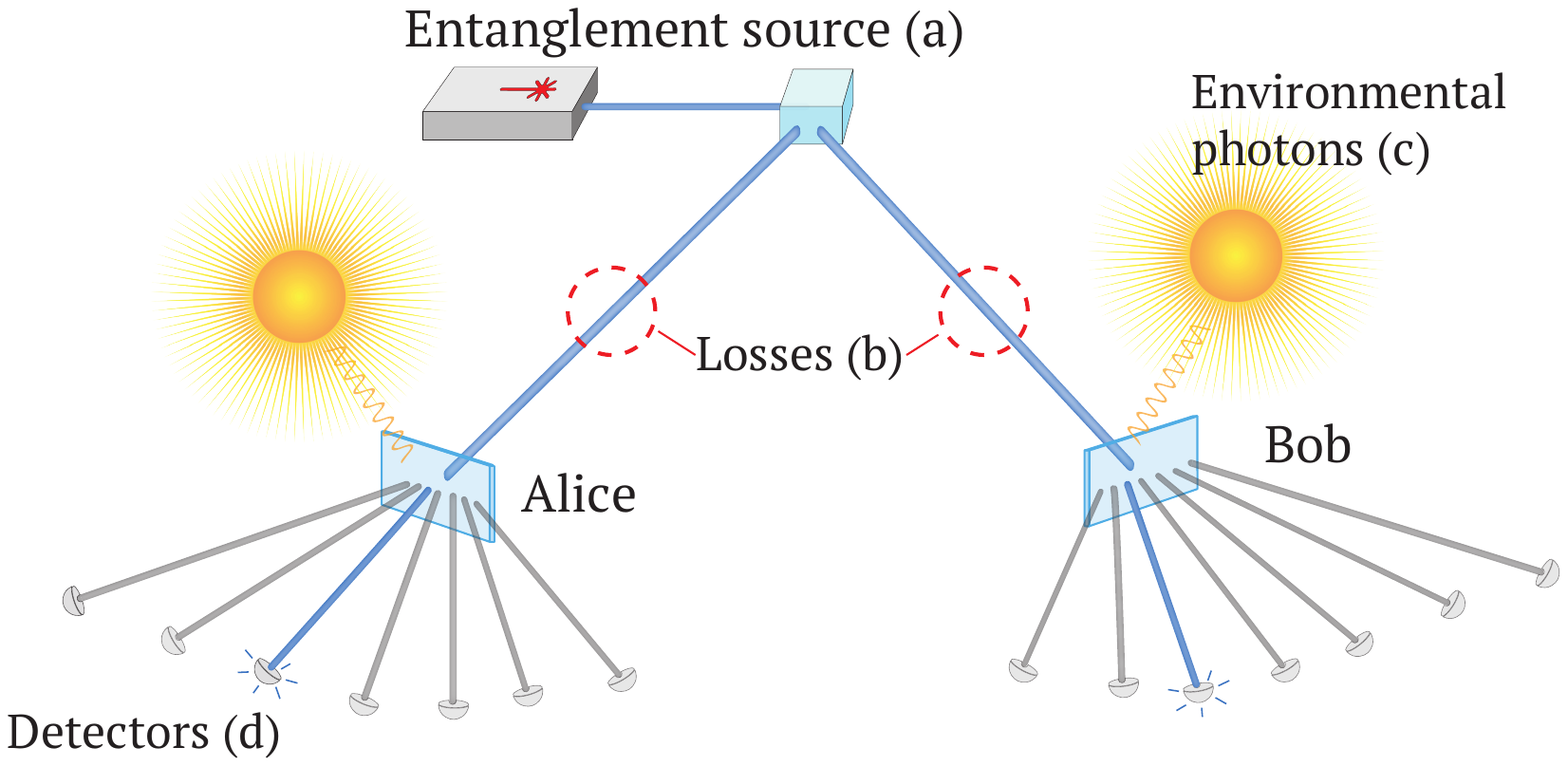}
     \caption{A schematic representation of the noise model: a laser source (a) produces entangled pairs distributed in time with a Poisson distribution, with $\lambda$ as the average number of photons per second. The pairs are distributed to the parties and suffer from party-dependent losses (b) with probability $P_L$. On top of the entangled photons the parties receive $\nu$ environmental photons per second on average, distributed as well with a Poissonian. For each of the modes that are being measured, there is an associated detector (d), which comes with an average number of dark counts $\mu$ per second and an efficiency $P_C$.}
     \label{fig:noise_model}
 \end{figure}

In our model, we consider that a click can come either from the laser or from the environment or from the dark counts. We start with the laser photons, by assuming that they follow a Poisson distribution, factorized by the probability $P_P(d)$ of being within the modes $-d/2$ and $d/2$. The probability that the laser produces $j$ detectable photons given that $n$ in total are emitted is
$ (\Delta t \lambda)^n \frac{e^{-\Delta t \lambda}}{n!} P_P^j(d) (1-P_P(d))^{n-j} \binom{n}{j}$. Given $j$ photons produced from the laser, the probability they produce no click in one of the labs is
\begin{equation*}
    P(0|j)=\sum_{r=0}^j P_L^{j-r} (1-P_L)^r \binom{j}{r}  (1-P_C)^{r}\binom{r}{0}
    =[1-P_C(1-P_L)]^j
    =(1-T)^j, \ \ \ \text{where}\ \ \ T=P_C(1-P_L),
    \end{equation*}
while the probability that they produce one or more clicks in a single detector is
\begin{align*}
    P(1|j)&=\sum_{r_1=1}^j P_L^{j-r_1} (1-P_L)^{r_1} \binom{j}{r_1}\sum_{r_2=1}^{r_1} \left(\tfrac{d-1}{d}\right)^{r_1-r_2} (1-P_C)^{r_1-r_2} \left(\tfrac{1}{d}\right)^{r_2} [1-(1-P_C)^{r_2}]\binom{d}{1}\binom{r_2}{r_1}\\&=d\bigg[\left(1-T+\tfrac{T}{d}\right)^j-(1-T)^j\bigg].
\end{align*}
In the above, $P_L$ reflects the losses affecting the entangled photons, and we further assumed that all modes suffer the same losses, therefore we absorbed them in $P_C$. In case one would like to further refine the noise model, they could consider different losses for different modes.
We also calculate the probability that, given $j$ photons were produced, Alice and Bob both get one click from a laser photon in different detectors, as this way we account for entangled photons. We have

\begin{align*}
   P(\neq|j)&=
   d(d-1)\sum_{\substack{r_1,r_2,r_3=0 \\ r_0+r_1+r_2+r_3=j}}^j\sum_{r_0=0}^{j-1}
   \frac{j!}{r_0! r_1! r_2! r_3!}
   \times
   (P_L^A P_L^B)^{r_0} [P_L^B (1-P_L^A)]^{r_1}  [P_L^A (1-P_L^B)]^{r_2} [(1-P_L^A)(1-P_L^B)]^{r_3} \times
   \\
   &
   \times
   \sum_{l_1=0}^{r_1} \left(\frac{d-1}{d}\right)^{r_1-l_1} (1-P_C^A)^{r_1-l_1} \left(\frac{1}{d}\right)^{l_1} \binom{r_1}{l_1}
   \times \sum_{l_2=0}^{r_2} \left(\frac{d-1}{d}\right)^{r_2-l_2} (1-P_C^B)^{r_2-l_2} \left(\frac{1}{d}\right)^{l_2} \binom{r_2}{l_2}
   \times
   \\
   &
   \times
   \sum_{\substack{s_3,p_3,q_3=0 \\ s_3+p_3+q_3=r_3}}^{r_3}\left(\frac{1}{d}\right)^{s_3+p_3} \left(\frac{d-2}{d}\right)^{q3} \frac{r_3!}{s_3!p_3!q_3!} (1-P_C^B)^{s_3+q_3}(1-P_C^A)^{p_3+q_3}
   \times[1-(1-P_C^A)^{l_1+s_3}][1-(1-P_C^B)^{l_2+p3}]
   = \\
   &=
   d(d-1)\Bigg\{[(1-T^A)(1-T^B)]^j+\bigg[(1-T^A)(1-T^B)+\frac{1}{d}[T^A(1-T^B)+ T^B(1-T^A)]\bigg]^j
   -\\
   &-\bigg[(1-T^B)(1-T^A+T^A/d)\bigg]^j - \bigg[(1-T^A)(1-T^B + T^B/d)\bigg]^j\Bigg\}.
\end{align*}
Similarly, the probability that, given $j$ photons, Alice and Bob both get one click from a laser photon in the same detector is
\begin{align*}
   P(=|j)
   &=
   d\sum_{\substack{r_1,r_2,r_3=0 \\ r_0+r_1+r_2+r_3=j}}^j\sum_{r_0=0}^{j-1}
   \frac{j!}{r_0! r_1! r_2! r_3!}
   \times
   (P_L^A P_L^B)^{r_0} [P_L^B (1-P_L^A)]^{r_1}  [P_L^A (1-P_L^B)]^{r_2} [(1-P_L^A)(1-P_L^B)]^{r_3} \times
   \\
   &
   \times
   \sum_{l_1=0}^{r_1} \left(\frac{d-1}{d}\right)^{r_1-l_1} (1-P_C^A)^{r_1-l_1} \left(\frac{1}{d}\right)^{l1} \binom{r_1}{l_1}
   \times
   \sum_{l_2=0}^{r_2} \left(\frac{d-1}{d}\right)^{r_2-l_2} (1-P_C^B)^{r_2-l_2} \left(\frac{1}{d}\right)^{l_2} \binom{r_2}{l_2}
   \times   \\
   &
   \times
   \sum_{l_3=0 }^{r_3} \left(\frac{d-1}{d}\right)^{r_3-l_3}  \left(\frac{1}{d}\right)^{l_3}  (1-P_C^B)^{r_3-l_3}(1-P_C^A)^{r_3-l_3} \binom{r_3}{l_3}
   \times
   [1-(1-P_C^A)^{l_1+l_3}][1-(1-P_C^B)^{l_2+l3}]=
   \\
   &=d\Bigg\{[(1-T^A)(1-T^B)]^j+\bigg[(1-T^A)(1-T^B)+\frac{1}{d}[T^A(1-T^B)+T^B(1-T^A)+T^A T^B] \bigg]^j
   -
   \\
   &
   -\bigg[(1-T^A+T^A/d)(1-T^B)\bigg]^j - \bigg[(1-T^A)(1-T^B+T^B/d)\bigg]^j\Bigg\},
\end{align*}

We can now proceed to the clicks due to dark counts.
Again, their distribution is Poissonian with multiple clicks in the same detector counting as one. Therefore, the probability of no clicks in a single detector is $e^{-\Delta t \mu}$, while the probability of one or more clicks in one detector is $1-e^{-\Delta t \mu}$.
In total, the probability of $n$ dark counts in all $d$ detectors is
$ P_D(n,d)= (e^{-\Delta t \mu})^{d-n}(1-e^{-\Delta t \mu})^{n} \binom{d}{n},$
which also gives another quantity that we need: the probability that, given that a detector already clicked because of a laser photon, all other detectors do not click because of dark counts. Denoting this probability by $P(0,d-1)$, we have \begin{equation*}
    P(0,d-1)=P_D(1,d)\frac{1}{d}+P_D(0,d)=e^{-(d-1)\Delta t \mu}.
\end{equation*}

Finally, we consider the last type of clicks that Alice and Bob register, the ones coming from environmental photons. We assume that they are produced according to a Poisson distribution.
Given $r$ photons in the same mode, the probability that at least one of them clicks is $1-(1-P_C)^r$. Furthermore, the probability that $r$ out of $q$ photons go in the same mode, one of them clicks, while all the others do not click is:
\begin{equation*}
 \sum_{r=1}^q \left(\frac{d-1}{d}\right)^{q-r} (1-P_C)^{q-r} \left(\frac{1}{d}\right)^{r}\big[1-(1-P_C)^r\big] {q \choose r} =\Bigg[1-\frac{P_C(d-1)}{d}\Bigg]^q -(1-P_C)^q,
 \end{equation*}
and we  multiply it with the Poissonian distribution of environmental photons and the number of modes to obtain
\begin{eqnarray*}
  P_E(1,d)=  d\sum_{q=0}^{\infty}(\nu \Delta t)^q\Bigg[1-\frac{P_C(d-1)}{d}\Bigg]^q \frac{e^{-\nu \Delta t}}{q!}-d\sum_{q=0}^{\infty}(\nu \Delta t)^q(1-P_C)^q \frac{e^{-\nu \Delta t}}{q!}
    =d P_E(0^*,d)\left(1-e^{-P_C \nu \Delta t/d}\right),
\end{eqnarray*}
with $P_E(0^*,d)=e^{-P_C \nu \Delta t(d-1)/d}$,
which is the probability that, in case a detector already clicked because of a laser photon or a dark count,
 $r$ out of $q$ environmental photons end up in this detector, while the rest $q-r$ end up in the other detectors and none of them clicks. Note that losses affecting environmental photons are absorbed in $\nu$.

\vspace{\baselineskip}
\noindent With all the above in place, we can now calculate the quantities of interest, namely the visibility and the key rate.
We start with the probability that, given $j$ photons locally, a single detector clicks
\begin{align*}
    P(1)&=P(1|j)P_D(0,d-1)P_E(0^*,d)+P(0|j)P_D(1,d)P_E(0^*,d)+P(0|j)P_D(0,d)P_E(1,d)
    \\
    &
    =d  P_D(0,d-1)P_E(0^*,d) \bigg[\left(1-T+\frac{T}{d}\right)^j-(1-T)^j e^{-\Delta t (\mu+P_C \nu/d)}\bigg],
\end{align*}
and we continue with the probability that both parties get a single click

\begin{align*}
    P(11)&=\sum_{n=0}^{\infty}\sum_{j=0}^n (\Delta t \lambda)^n \frac{e^{-\Delta t \lambda}}{n!} P_P^j(d) [1-P_P(d)]^{n-j} \binom{n}{j} \\
    &\times\big\{
    [P(1|j)P_D(0,d-1)P_E(0^*,d)+P(0|n)P_D(1,d)P_E(0^*,d)+P(0|n)P_D(0,d)P_E(1,d)]^A
    \times\\
   & \times
   [P(1|j)P_D(0,d-1)P_E(0^*,d)+P(0|n)P_D(1,d)P_E(0^*,d)+P(0|n)P_D(0,d)P_E(1,d)]^B+\\
   &+[P_D(0,d-1)P_E(0^*,d)]^A[P_D(0,d-1)P_E(0^*,d)]^B [P(\neq|j)+P(=|j)-P^A(1|j)P^B(1|j)]\big\}=
   \\
   \hspace{0.1cm}
     &=
     d e^{-\Delta t (d-1)(\mu^A+\xi^A/d+\mu^B+\xi^B/d)} e^{-\Delta t \gamma} \Bigg\{d\left(1-e^{-\Delta t (\mu^A+\xi^A/d)}\right) \left(1-e^{-\Delta t (\mu^B+\xi^B/d)}\right) + e^{\Delta t \gamma/d} -1\Bigg\},
    \end{align*}
where
\begin{equation*}
    \xi^{A/B}=P_C^{A/B} \nu^{A/B}  + P_P(d) \lambda P_C^{B/A}(1-P_L^{B/A})(1-P_C^{A/B}+P_C^{A/B} P_L^{A/B})\ \ \ \text{and}\ \ \
\end{equation*}
are all experimental constants independent of $d$. Note that $\gamma$ is the same as in the previous implementation of temporal encoding, and represents the average number of detectable entangled photons, while $\xi$  represents the environmental and laser photons that click independently in the labs. We are now able to write that the rate of ``valid'' rounds per second is
\begin{equation*}
   R(d)=\frac{P(11)}{\Delta t}= C  e^{-d(\mu^A+\mu^B)}e^{(\xi^A+\xi^B)/d}  \Bigg\{d^2\Big[1-e^{- (\mu^A+\xi^A/d)}\Big] \Big[1-e^{-(\mu^B+\xi^B/d)}\Big] +d\left( e^{ \gamma/d} -1\right)\Bigg\},
\end{equation*}
where $C=e^{\Delta t (\mu^A+\mu^B-\xi^B-\xi^B -\gamma)}/\Delta t$.\\

Finally, in order to get the expression for the visibility, we need the probability that an entangled pair clicks on both labs, while all other detectors do not click. However, once the detectors click because of the entangled pair, they might also receive any number of other photons and register dark counts. We can go around this cumbersome calculations, by directly considering the probability that the same detector clicks for both Alice and Bob (which is due to the correlated photons and the noise), and subtract the probability that different detectors click (which is due to the noise only). Essentially, we  subtract the $P(\neq|j)$ contribution in $P(11)$ from its $P(=|j)$ contribution to obtain
\begin{equation*}
    P_S=d e^{-\Delta t (d-1)(\mu^A+\xi^A/d+\mu^B+\xi^B/d)} e^{-\Delta t \gamma} \left(
 e^{\Delta t \gamma/d} -1\right),
\end{equation*}
which, in turn, gives us the visibility
\begin{equation*}
    v(d)=\frac{P_S}{P(11)}=\frac{1}{1+d\left(1-e^{-\Delta t (\mu^A+\xi^A/d)}\right) \left(1-e^{-\Delta t (\mu^B+\xi^B/d)}\right) (e^{\Delta t \gamma/d} -1)^{-1}}.
\end{equation*}
By re-scaling with $\Delta t$, we can also write
\begin{equation*}
v(d)=\frac{e^{ \gamma/d} -1}{e^{ \gamma/d} -1+d\big[1-e^{- (\mu^A+\xi^A/d)}\big] \big[1-e^{-(\mu^B+\xi^B/d)}\big] }. \end{equation*}
Note that for large $d$ and small $\Delta t$ the visibility scales as
$v(d) \approx \frac{1}{1+d^2 \Delta t \mu^A \mu^B \gamma^{-1}}.$

\end{document}